\documentclass[11pt]{llncs}

\usepackage{epsfig}

\usepackage{amsmath}
\usepackage{amssymb}

\begin{document}

%\title{Mathematical models for the emergence of highly organized structures in evolving societies}

\title{On the existence of highly organized
communities in networks of locally interacting agents}

%\institute{Computer Technology Institute and Press --
%``Diophantus'', Patras, 26504, Greece \and Dept. of Business Administration, University of Patras, Patras, Greece.
%\\ \email{emails: }}

\author{V. Liagkou \inst{1,4}
\and P.E. Nastou \inst{5,6} \and P. Spirakis \inst{2} \and Y.C. Stamatiou\inst{1,3}}

\institute{Computer Technology Institute and Press - ``Diophantus'', University of Patras Campus, 26504, Greece \and
Department of Computer Science, University of Liverpool,
UK and Computer Engineering and Informatics Department, University of Patras, 26504, Greece \and
Department of Business Administration, University of Patras, 26504, Greece \and
University of Ioannina, Department of Informatics and Telecommunications, 47100 Kostakioi Arta, Greece  \and
Department of Mathematics, University of the Aegean, Applied Mathematics and Mathematical Modeling Laboratory, Samos, Greece
\and
Center for Applied Optimization, University of Florida, Gainesville, USA
\\ e-mails: liagkou@cti.gr, pnastou@aegean.gr, P.Spirakis@liverpool.ac.uk, stamatiu@ceid.upatras.gr}

\maketitle

\begin{center}
\today
\end{center}

\begin{abstract}
In this paper we investigate
phenomena of spontaneous emergence or purposeful formation
of highly organized structures in networks of related agents.
We show that the formation of large organized structures
requires exponentially large, in the size of the structures, networks.
Our approach is based on Kolmogorov, or descriptional, complexity
of networks viewed as finite size strings.
% One of the main results
%of this theory is that random, in the Kolmogorov complexity context,
%strings (i.e. networks or, in our model, graphs) cannot contain
%large regular substrings (subgraphs or subnetworks modeling organized structures with, assumed regularities).
%
We apply this approach to the study of the emergence or formation of
simple organized, hierarchical, structures based on Sierpinski Graphs and
we prove a Ramsey type theorem that bounds the number of vertices in Kolmogorov random graphs that contain Sierpinski Graphs as subgraps.
Moreover, we show that Sierpinski Graphs encompass close-knit relationships among their vertices
that facilitate fast spread and learning of information when
agents in their vertices are engaged in pairwise interactions modelled as two person games. Finally, we generalize our findings for any organized structure with succinct representations.
Our work can be deployed, in particular, to study problems related to the security of networks by identifying conditions which enable or forbid the formation
of sufficiently large insider-subnetworks with malicious common goal to overtake the network or cause disruption of its operation.
\end{abstract}


\section{Introduction}

In this paper, we address the problem of the formation or
emergence of certain configurations or patterns,
termed {\em highly organized structures} in our context,
in evolving networks of agents (e.g. social networks)
based on combining the concepts of
{\em Ramsey Numbers in graphs} and {\em Kolmogorov Complexity}.
Our work does not rely on random graph or asymptotics techniques and focuses on finite graphs
and can be used to study when and how organized structures appear in finite size networks of interacting agents, answering
questions such as the following: (i) given a network of interacting agents of a specific size,
how large can an organized subnetwork be in which a common goal diffuses
to overtake the whole network?
(ii) can a subgroup of a specific organization structure appear in networks of a given size? and
(iii) how large a network of interacting agent should be before it becomes vulnerable to the possibility of formation
of large subnetworks of closely-interacting malicious agents?

Broadly speaking, {\em Ramsey theory} refers to mathematical statements
that a specific structure (e.g. string, graph, number sequence etc.) is certain to contain a large,
highly organised, substructure. Examples of such statements stems from several disciplines,
including mathematical logic, number theory, analysis, and graph theory. One of the most
well-known such statements in number theory, proved by van der Waerden in 1927
(see~\cite{vdW27}) is the following: for {\em any} fixed positive integers $r$ and $k$,
there exists some positive integer $n$ such that if the integers $\{1, 2, \ldots, n\}$
are colored using $r$ different colors, then there exist {\em at least} $k$ monochromatic
integers in arithmetic progression. The {\em least} such value of
$n$ is called the {\em Van der Waerden number} $W(r, k)$.
In this paper, our focus is on the Ramsey Theory of {\em graphs}. This research
area was inaugurated in 1930 by Ramsey in~\cite{R30} in which he
stated and proved the, so called, {\em Ramsey's Theorem}:
for {\em any} graph $H$, there exists a natural number $n$ such
that for {\em any} colouring of the edges of $K_n$ (i.e. the clique on $n$ vertices)
with two colours, $K_n$ contains a {\em monochromatic} copy
of $H$ as a subgraph, not necessarily induced.
The {\em least} such value $n$ is called the {\it Ramsey number} of $H$
and is denoted by $r(H)$ (if $H = K_t$, we simply write $r(t)$).
For our purposes, we consider the Ramsey numbers for the class of
{\em Sierpinski Graphs}. This
class of graphs has {\em bounded} maximum degree for its vertices. This fact leads to
{\em linear}, in the number of vertices of $H$, upper bounds for the Ramsey numbers of these graphs. This
follows from a more general result of Chv\'atal, R\"odl, Szemer\'edi and Trotter in \cite{CRST83} which
states that if $H$ is a graph on $n$ vertices and maximum degree
$\Delta$, then the Ramsey number $r(H)$ is bounded by $c(\Delta) n$
for some constant $c(\Delta)$ depending only on $\Delta$.
That is, the Ramsey number of bounded-degree graphs grows linearly in the number of vertices.
Moreover, a linear bound also holds for the {\em induced}
Ramsey numbers which are defined as the Ramsey numbers only, now,
the monochromatic copy of $H$ needs to form an {\em induced} subgraph,
which is a stricter form of subgraph that allows
the existence of edges in the subgraph $H$ if and only if they are edges of $H$.
Other edges, which may exist among vertices of $H$ as edges of the larger graph
should not exist, as it may be the case in the general subgraph notion.
The {\em Induced Ramsey Theorem} states that
$r_{\textrm{ind}}(H)$ exists for every graph $H$ (see, e.g., Chapter 9.3 in \cite{Di}).

On the other hand, the {\em Kolmogorov complexity} of a {\em finite object}
(most often a bit sequence modelling a finite structure such as a graph)
is defined as the {\em minimum} number of
bits into which the object can be compressed without losing information, i.e. so as the
compressed string is recoverable through some algorithm, running on
a reference machine, which is most often the universal Turing machine (see~\cite{LiVi19}).
Note the the word {\em complex} in the context of
Kolmogorov complexity reflects, rather, {\em lack} of organization,
pattern, and ``sophistication'' as opposed to the meaning
the word complex has in everyday language as describing
an object or entity of a high degree of organization and fine
structural details. Thus, an object of {\em high} Kolmogorov complexity
is, actually, of low complexity in the usual sense since they have all
{\em randomness} properties while objects
of {\em low} Kolmogorov complexity lack randomness, having
internal structural organization and many regularities.
Thus, they can be {\em succinctly} described or {\em compressed},
in the terminology of Kolmogorov complexity theory.
With regard to this concept of complexity, it was proposed and developed,
independently, by Andrei Kolmogorov, Ray Solomonoff, and Gregory Chaitin
(\cite{kolm,sol,chaitin}). This concept addresses the
complexity of {\em finite objects}, in contrast with classical
complexity theory, like NP-completeness, that addresses the complexity
of {\em infinite} sets of finite objects, i.e.
languages in the complexity-theoretical terminology.

Based on Ramsey theory and Kolmogorov Complexity arguments,
we investigate, in this paper,
phenomena of existence or formation of highly organized structures,
such as organizations, people's networks, and societal patterns, in evolving networks of agents.
As a specific case, we focus on the class of {\em Sierpinski Graphs} which have, also, the important property
of being {\em close-knit}, i.e. agents interconnected with connections based on these graphs
are very ``cohesive''. Close-knit graph families
play an important role in diffusion processes in graphs
(see~\cite{young}) in contexts such as the emergence of shared and agreed upon
innovations or ideas over large population of agents..

%We expect our work to contribute to the study of phenomena in Social Sciences using
%formal methods and results, mainly from mathematics, as proposed, e.g.,
%by Kauffman in~\cite{Kauffman93,Kauffman95}. This approach relies on
%using appropriate mathematical societal models and study methods
%whose abstraction power and representation
%simplicity renders complex societal phenomena amenable to formal treatment and
%formal derivation of results.

\section{Closeness properties of graphs}

%Our focus in this paper is to contribute to the formal study of societal structures and
%their properties by retaining only some of their most fundamental elements and characteristics:
%(i) entities, (ii) relations and interactions between entities,
%(iii) evolutionary processes, (iv) emergence of organized societal patterns.
%
%To this end, we use the {\em graph} combinatorial structure to model societies' entities,
%relations, and interactions as well as Kolmogorov complexity and {\em Ramsey theory} to
%study the emergence of properties and organizes structures in large networks.
%
Agent interactions and the study of their long term evolution as well as properties using such models, can reveal crucial information about the dynamics of the social fabric and the way it expands as well as how {\em organized} substructures emerge, {\em inescapably}, due to these dynamics.
For instance, in this paper we show that, under certain conditions, one may
locate in large structures containing interacting agents, large substructures with certain ``member closeness'' properties. We show, for instance, in this paper that large {\em close-knit} regular substructures emerge, as the network
expands, as a form of {local ``island''} of regularity and social organization, even if the network of agents evolves in {\em unpredictable} (i.e. random) ways on a global scale. Close-knit agent structures, as Young showed in~\cite{young}, play an important role in {\em diffusing} throughout
society ideas and beliefs that are, initially, held by only a relatively small group of interacting agents.

Young's mathematical result examines how an idea or opinion spreads over a population of interacting agents whose social structure has a certain ``coherence'' property and whose members interact through a two person ``innovation acceptance'' game.

In~\cite{young}, Young defined a parameter of networks of agents, modelled as {\em graphs}, that describes their ``coherence'' as well as their vulnerability to outsiders' views against the innovation to be diffused.
A group of agents $S$ (i.e. nodes in a given graph $G$ with no isolated vertices) is {\em close-knit} if the following condition is true
for every $S' \subseteq S, S'\neq \emptyset$, given appropriate values for $r$ (see discussion below):
\begin{equation}
\min_{S' \subseteq S} \frac{d(S',S)}{\sum_{i \in S'} d_i} \geq r
\label{youngineq}
\end{equation}
where $d(S',S)$ is the internal connections of $S'$ i.e. the number of links $\{i,j\}$ where $i\in S'$ and $j\in S$ while $d_i$ is the total number
of links that agent $i$ possesses (i.e. its degree in the graph). It is obvious from Inequality~(\ref{youngineq}) that for every subset of a {\em close-Knit}
group $S$, the ratio of internal degree to the total degree is at least $r$.
Intuitively, to have a large such ratio in a group of agents we need, relatively, many internal connections and few external connections.
%Intuitively, to have a large such ratio in a group of agents we need, relatively, many internal connections and few external connections.
%
%This parameter encompasses
%internal connection strength (fraction numerator) and external
%connection strength (fraction denominator) of the set of individuals $S$.
%

Given a positive integer $k$ and a real number $ 0 \leq r \leq 1/2$,
we call a graph $G$$(r,k)$-close-knit if every
agent belongs to a group $S$
of cardinality at most $k$ which is $r$-close-knit
as defined by Inequality~(\ref{youngineq}).
Finally, a {\em class} of graphs is {\em close-knit}
if for every $0 \leq r \leq 1/2$ there exists an integer $k$, possibly
depending on $r$, such that every graph in the class is $(r,k)$-close-knit.

Moreover, the agents of the population are involved in playing a two person
game on a regular basis in which pairs of interacting agents compete for a payoff against the other (non zero-sum game).
According to the game, each time two agents interact they, independently, choose to either adopt or not adopt the innovation, based on the payoffs they can receive from their choices.
The game is assumed to have a {\em risk dominant} Nash equilibrium in which
{\em both} players adopt the innovation (see~\cite{young} for the details).

Based on these two elements,
Young proved that for classes of close-knit agent graphs
{\em all} community members will eventually adopt the
innovation, i.e. the risk dominant
Nash equilibrium of the game, in a number
of interactions (``time'') which is bounded and independent
from the size of the community.
In other words, if all the subsets of the community members have
strong pairwise links and, at the same time, are weakly connected to
outsiders (who may even be negative towards adopting the
innovation), then a group of initiators will manage, in the end, to
convince all population members to adopt the innovation.

%%%%%%%%%%%%%%%%%%%%%%%%%%%%%%%% CAN BE NEW %%%%
\section{Sierpinski triangle based group formations}
\label{sierpinski-networks}

In this section we describe a class of graphs based on the {\em Sierpinski Triangle}
fractal and prove that it forms a {\em close-knit graph family}.
%For more on these properties, based on Young's theory, please see our previous work~\cite{MSE19}. For
%completeness, below we summarize basic definitions and results from this work that will be used later on in this paper.
%

\begin{definition}[Sierpinski Triangle gasket of level $l$]
Given an integer $l$, $l \geq 1$, we define the Sierpinski Triangle gasket of level $l$
or, simply, Sierpinski Triangle of level $l$ as follows: for $l = 1$ the Sierpinski Triangle is
an equilateral triangle while for $l > 1$, the Sierpinski Triangle of level $l$ is
composed of three copies of a Sierpinski Triangle of level $l - 1$ connected at their corners.
\end{definition}
In Figure~\ref{pattern1} we see a few of the first Sierpinski Triangles, for $l = 1, 2, 3, 4, 5, 6$  and $7$.
\begin{figure}[htb]
\centerline{\epsfig{file=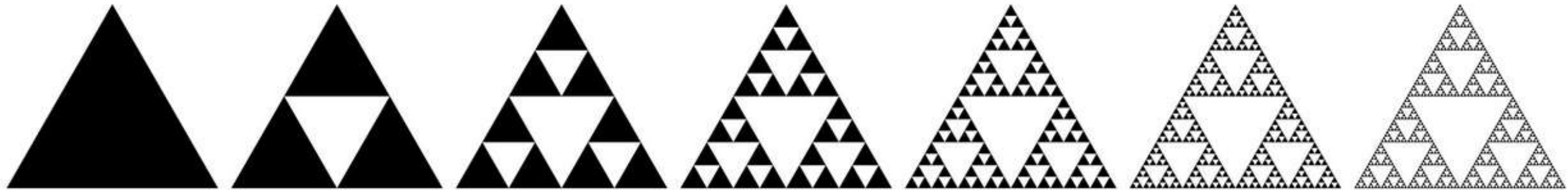,width=12cm}}
\caption{Sierpinski Triangles, for $l = 1, 2, 3, 4, 5, 6$  and $7$ respectively.}
\label{pattern1}
\end{figure}
Based on the Sierpinski Triangles, we can define a corresponding family of graphs, called {\em Sierpinski Graphs}.
\begin{definition}[Sierpinski graphs]\label{defKOPG}
The Sierpinski Graph of level $l$, $l \geq 1$, denoted by $S_l$ is formed as follows:
if $l = 1$ then the Sierpinski Graph of level 1 is formed if we replace the
three vertices and the edges of the Sierpinski Triangle of level 1 with graph vertices and edges
otherwise, for $l > 1$, the Sierpinski Graph of level 1 is formed by
three copies of the Sierpinkski Graph of level $l - 1$ by identifying their
vertices corresponding to the corners of the corresponding Sierpinski Triangle of level $l - 1 $.
\end{definition}

In Figure~\ref{pattern2}, we see the Sierpinski Graphs $S_1, S_2$ and $S_3$.
\begin{figure}[htb]
\centerline{\epsfig{file=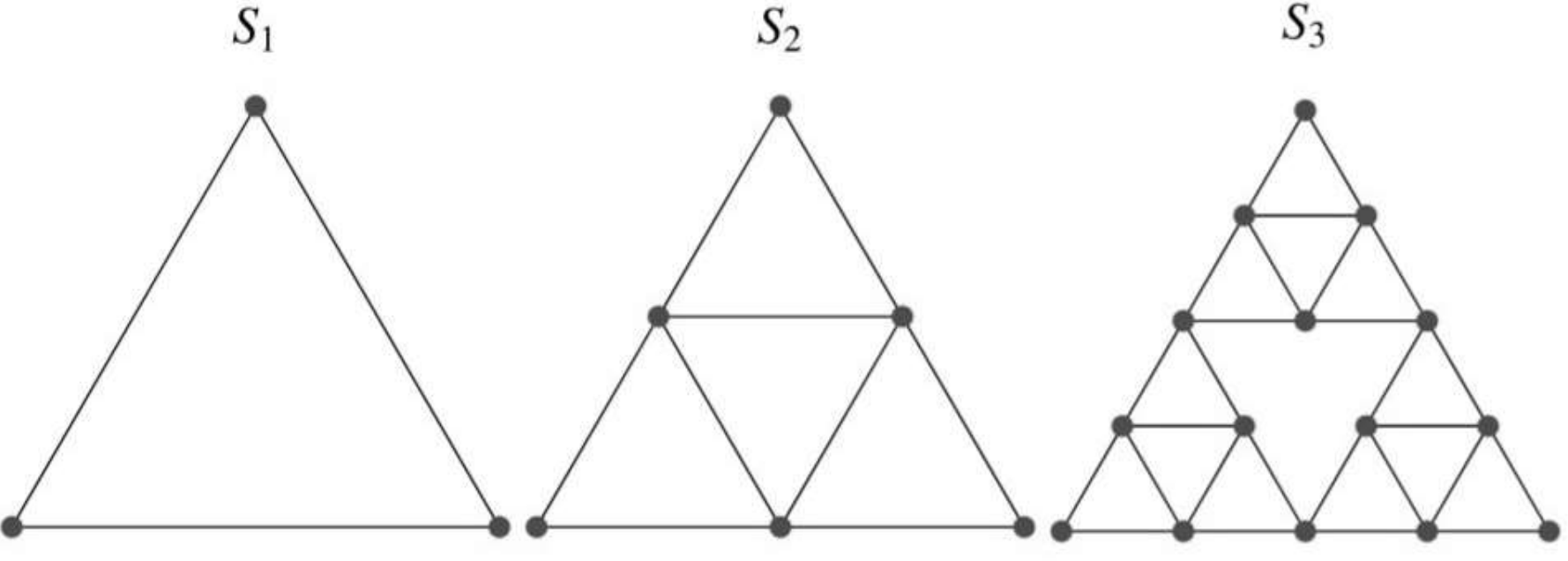,width=6cm}}
\caption{The Sierpinski Graphs $S_1, S_2$ and $S_3$.}
\label{pattern2}
\end{figure}
%Bodlaender has shown that every $k$-outerplanar graph has treewidth at most
%$3k-1$ \cite[Theorem 83]{Bod98}. In our case, the outerplanar graphs correspond to $k = 1$ which implies
%that their treewidth is at most $2$. This, according to~\cite{ABD07}, leads
%to the fact that local consistency methods such as arc-consistency discussed in Section,
%produce social networks on which global consistency can be imposed using a fast
%value assignment algorithm without backtracking, which is a highly inefficient technique in distributed
%computing networks such
%
%\begin{theorem}[Chartrand and Geller~\cite{Chartrand1969}]
%\label{thm:outerplanar-unique3}
%An outerplanar graph $G$ with at least three vertices is uniquely 3-colorable if and only if it is maximal outerplanar.
%\end{theorem}
%
In particular, in Figure~\ref{pattern3} we see how the
Sierpinski graph of level $4$, $S_4$, is composed
of three copies of the Sierpinski Graph of
level 3, $S_3$, which is shown in a dashed enclosure.
\begin{figure}[htb]
\centerline{\epsfig{file=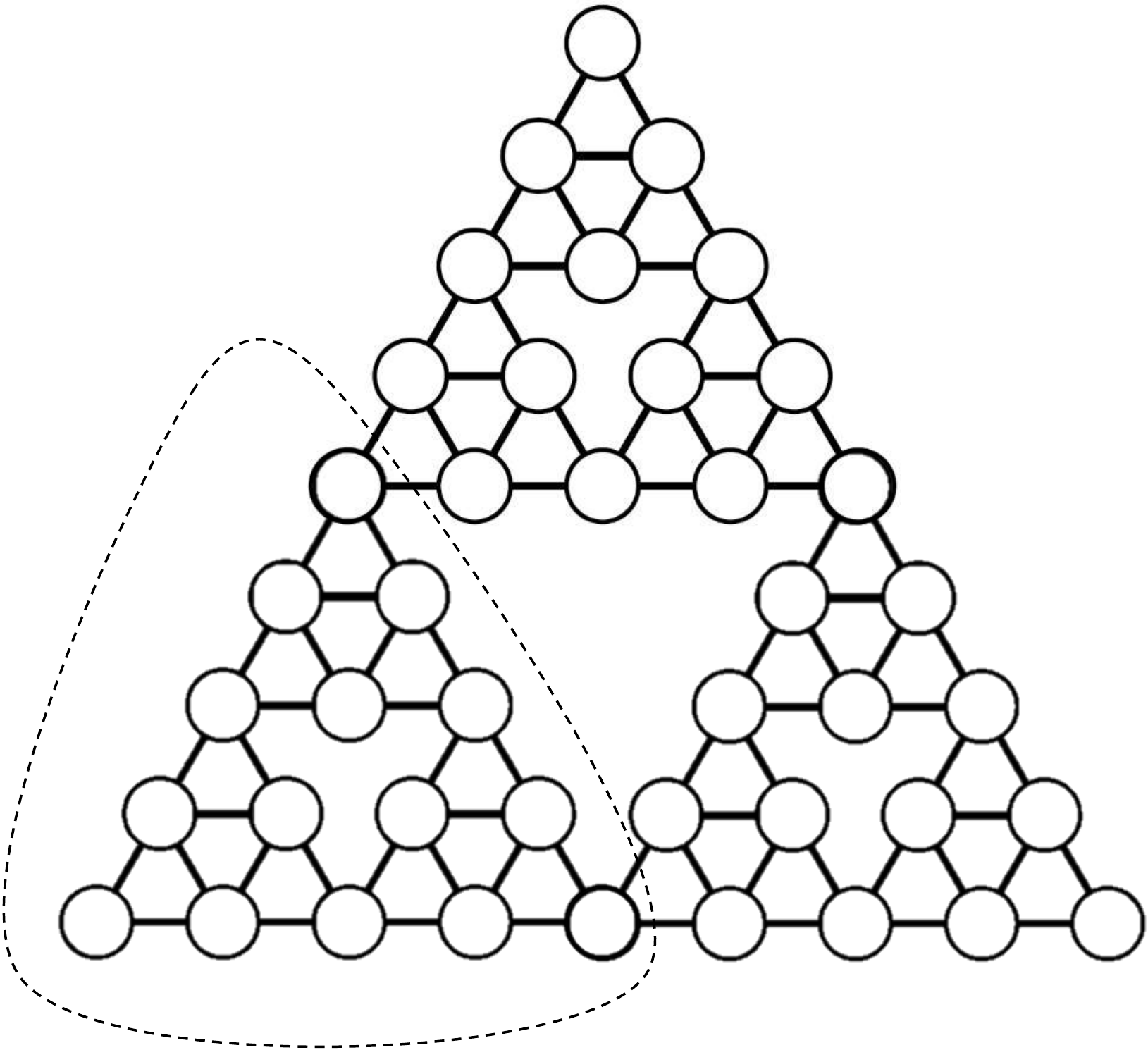,width=5cm}}
\caption{The Sierpinski Graph of level 4.}
\label{pattern3}
\end{figure}
%http://graphonline.ru/en/?graph=fXQATxIVXQqGqLqU
%
It is not hard to see that the following properties hold for $S_l$:
\begin{itemize}
\item For every value of $l > 1$,
all vertices of $S_l$ have degree 4 except three vertices (the ``corner'' ones) which have degree 2.

\item The number of vertices of $S_l$ is $n_l = \frac{3}{2} (3^{l - 1} + 1)$.

\item The number of edges of $S_l$ is $m_l = 3^l$.

\end{itemize}
The degree and connectivity properties
helps satisfy Inequality~(\ref{youngineq}). Intuitively, having too many nodes with large degrees,
exposes a group of agents to much external interference. This affects negatively the close-knit property.
%
%In~\cite{MSE19}, the following was proved:
We prove the following:
\begin{theorem}
The family of Sierpinski Graphs is close-knit.
\end{theorem}
\begin{proof}
Let $S_l = (V,E)$ be a Sierepinski Graph of level $l$, $ l \geq 1 $. Let $S$ be a subset of its set of vertices $V$, such that $|S| = k$ and $S'$ a nonempty subset of $S$. Each of the vertices in $S$ and $S'$ have degree either $2$ or $4$. Let, also, $i$, $0 \leq i \leq 3$ be the number of vertices of degree $2$ in $S'$. We compute the numerator and denominator of the fraction in Inequality~(\ref{youngineq}).

%Let $S = V$, i.e. all the vertices in $S_l$ and $S' \subseteq S$ any subset of its vertices of cardinality $k$, i.e. $|S'| = k$, for $k \geq 1$.
%The vertices in $S'$ can either be of degree 2 or 4. Let $i$ be the number of vertices of degree 2. Then, $0 \leq i \leq 3$.
If by $\alpha$, $0 \leq \alpha \leq 1$, we denote the fraction of edges between vertices in $S'$ and vertices in $S - S'$, then with respect to the numerator of the fraction in Inequality~(\ref{youngineq}), we have
\begin{eqnarray}
d(S', S) &=& \alpha [4(k-i)+2i] + (1-\alpha)\ \frac{4(k-i)+2i}{2}\\
		&=& \alpha [4(k-i)+2i] + (1-\alpha)\ [2(k-i)+i].
\label{numer}
\end{eqnarray}
since edges between vertices in $S'$ should be counted only once for the calculation of $d(S',S)$ in~(\ref{numer}). As for the denominator in Inequality~(\ref{youngineq}), we have
\begin{equation}
\sum_{i \in S'} d_i = 4(k-i) + 2i.
\label{denom}
\end{equation}
From~(\ref{numer}) and~(\ref{denom}) we have
\begin{equation}
\frac{d(S',S)}{\sum_{i \in S'} d_i} = \frac{\alpha[4(k-i)+ 2i]+(1-\alpha)\ [2(k-i) +i]}{4(k-i) + 2i}.
\label{fraction}
\end{equation}
After some straightforward manipulations,~(\ref{fraction}) reduces to $\frac{1 + \alpha}{2}$.

Therefore, Inequality~(\ref{youngineq})
holds for any $r$ and $k$. We conclude, that $S_l$, for $l \geq 1$,
i.e. the Sierpinski Graphs, form a close-knit family of graphs.
\hfill $\square$
\end{proof}

\section{Existence of highly organized structures in networks of interacting agents}
\label{KC}
In this section, we investigate the conditions upon which highly organized
structures can appear in networks of interacting agents. These structures may represent any organized community of individuals with a common aim or shared beliefs. These conditions, loosely speaking, rely on the structures possessing certain recognizable {\em regularity properties} that can be used to {\em succintly} describe them. These {\em regularity} and {\em succinctness} conditions model the ``{\em organized structure}'' notion.
For concreteness, we will fix our organized, regular, structures to belong in the class of {\em Sierpinski Graphs} which, in addition, have the property of being {\em close-knit} (see Section~\ref{sierpinski-networks}).

More specifically, we will study the problem of the existence of Sierpinski Graphs, as subgraphs, in sufficiently large graphs, which can model evolving interacting agents, social networks and societies.
In this section we deploy techniques from Kolmogorov Complexity and Ramsey Theory. In what follows, we briefly state the main definitions and some useful results from {\em Kolmogorov Complexity} and {\em Ramsey Theory}.
We, then, apply both theories in order to investigate conditions that enable or hinder the emergence of organized substructures in evolving structures which are modelled as graphs.

\subsection{Kolmogorov Complexity}

%The statistical properties of strings with high Kolmogorov complexity
%have been studied in the seminal book~\cite{LiVi19}, on which this section is based.
In informal terms (see~\cite{kolm,sol,chaitin}) the {\em Kolmogorov Complexity} of a (binary) string $x$ is the {\em length} of the {\em shortest algorithmic description} of $x$.
In other words, the {\em Kolmogorov Complexity}, denoted by $C(x)$, of
a finite string $x$ is the length of the {\em shortest program} (or Turing machine in general) encoding as a binary sequence of bits, which produces $x$ as output, without taking any input. Similarly, the {\em conditional} Kolmogorov Complexity of $x$ given $y$, denoted by $C(x|y)$, is the length of the shortest program which produces $x$ as output given $y$ as input.
It can be shown that $C(x)$ is, in some sense, {\em universal} in that it
does not depend on the choice of the programming language or Turing machine model, up to  fixed additive constant, which depends on this choice but not on $x$.

In this paper, our focus is on graphs. We can deploy the notion and properties of Kolmogorov Complexity by encoding graphs with strings as follows
(see, e.g.,~\cite{LiVi19}):
\begin{definition}
\label{def.gc}
Each labelled graph $G=(V,E)$ on $n$ nodes $V=\{1,2,\ldots, n\}$ can be represented (up to automorphism) by a binary string $E(G)$ of length ${n \choose 2}$. We simply assume a fixed ordering of the ${n \choose 2}$ possible edges in an $n$-node graph, e.g. lexicographically, and
let the $i$th bit in the string indicate presence (1) or absence (0) of
the $i$th edge. Conversely, each binary string of length ${n \choose 2}$ encodes an $n$-node graph. Hence we can identify each such graph with
its binary string representation.
\end{definition}
\begin{definition}\label{def.rg}
A labelled graph $G$ on $n$ nodes has {\em randomness deficiency}\index{randomness deficiency} at most $\delta (n)$, and is called
$\delta (n)$-{\em random}, if it satisfies
\begin{equation}\label{eq.KG}
C(E(G)|n ) \geq {n \choose 2} - \delta (n).
\end{equation}
\end{definition}
Also, the following holds (see, e.g.,~\cite{LiVi19}):
\begin{lemma}\label{lem.frac}
A fraction of at least
$1 - 1/2^{\delta (n)}$ of all labelled  graphs $G$ on $n$ nodes is $\delta (n)$-random.
In particular, for $\delta(n) = \log n$, a fraction
of $(1 - \frac{1}{n})$ of all graphs on $n$ vertices is $\log n$-random.
\label{def-lemma}
\end{lemma}

\begin{definition}
Let $G=(V,E)$ be a labelled graph on $n$ nodes. Consider a labelled graph $H$ on $k$ nodes $\{1,2, \ldots,k\}$. Each subset of $k$ nodes of $G$ induces a subgraph $G_k$ of $G$. The subgraph $G_k$ is an  ordered labelled {\em occurrence} of $H$  when we obtain $H$ by relabelling the nodes $i_1< i_2<  \cdots < i_k$ of $G_k$ as $1,2, \ldots, k$. 
\end{definition}
%
%Let $\#H(G)$ be {\em the number of times $H$ occurs}
%as an ordered labelled subgraph of $G$ (possibly overlapping). Let 
%$p$ be the probability that we obtain $H$ by flipping a fair
%coin to decide for each pair of nodes whether
%it is connected by an edge or not,
%\begin{equation}\label{eq.defp}
%p=2^{-k(k-1)/2}.
%\end{equation}
%\begin{theorem}\label{theo.freqG}
%Assume the terminology above with $G=(V,E)$ a labelled graph on $n$ nodes, $k$ is a positive integer
%dividing $n$, and
%$H$ is a labelled graph on $k \leq \sqrt{2 \log n}$ nodes. Let
%$C(E(G)|n ) \geq {n \choose 2} - \delta (n)$. Then
%\[ \left|\#H(G)- {n \choose k}p \right| \leq 
% {n \choose k} \sqrt{\alpha (k/n) p} , \]
%with
%$\alpha := (K(H|n) + \delta(n) + \log {n \choose k}/(n/k) + O(1))
%3 / \log e$. 
%\end{theorem}

\subsection{Ramsey Theory}
We first provide the definition of the $r$-colouring of a set in Ramsey Theory. 
\begin{definition} ($r$-colouring)
Let $S$ be a set and $r\in Z_{+}$. An $r$-colouring of $S$ is a function $f:S\rightarrow \{1,2,\ldots,r\}$.
\end{definition}	

As discussed in the introduction, Ramsey Theory is concerned with questions involving the appearance of certain patterns in sufficiently large graphs. 
For instance, Ramsey Theory started with the following question: given a graph $H$, determine the Ramsey number $r(H)$, which is defined as the
{\em smallest} natural number $n$ such that {\em any} two-colouring of $E(K_n)$ contains a monochromatic copy of $H$.
In~\cite{CFS12} the following is proved:
\begin{theorem} \label{thm:CRSTBound}
	There exists a constant $c$ such that any graph $H$
	on $k$ vertices with maximum degree $\Delta$ satisfies
	\[r(H) \leq k\ 2^{c \Delta \log \Delta}.\]
\end{theorem}
In this paper, we are interested in the {\em Induced Ramsey Number}
of a given graph $H$, which is defined as follows:

\begin{definition} (Induced Ramsey numbers)
A graph $H$ is an {\it induced subgraph} of a graph $H$ if $V(H) \subset
V(G)$ and two vertices of $H$ are adjacent if and only if they are adjacent in
$G$. The {\it induced Ramsey number} $r_{\textrm{ind}} (H)$ is defined as the {\em minimum} integer
for which there is a graph $G$ on $r_{\textrm{ind}} (H)$ vertices such that
every two-colouring of the edges of $G$ contains an induced monochromatic copy of $H$.
\label{def-induced-ramsey}
\end{definition}
Note that an induced monochromatic copy of $H$ is, also,
an {\em induced} copy of $H$ as an induced subgraph,
in the ordinary graph-theoretical sense, regardless of the colour of the edges of $G$.

There is a number of results that provide upper bounds on induced
Ramsey numbers for sparse graphs.
For instance, Beck in~\cite{Be}
focused on the case when $H$ is a tree. Also, Haxell, Kohayakawa, and \L
uczak \cite{HaKoLu} showed that the cycle of length $k$ has induced
Ramsey number linear in $k$. Moreover,
\L uczak and R\"odl \cite{LuRo} proved that the induced Ramsey
number of a graph with bounded maximum degree is at most polynomial in the
number of its vertices, settling a conjecture of Trotter.
More precisely, they proved  the following:
\begin{theorem}
For every integer $d$, there is a constant $c_d$ such that every graph $H$ on
$k$ vertices and maximum degree at most $d$ satisfies
$r_{\textrm{ind}}(H) \leq k^{c_d}$.
\label{upper-induced}
\end{theorem}
The proof  provides an upper bound on $c_d$ that is a tower of $2$'s of height proportional to $d^2$.
Since a Sierpinski Graph has maximum degree equal to 4,
then an immediate corollary from Theorem~\ref{upper-induced},
applied for $d = 4$, is the following:
\begin{corollary}
Let $H = S_l$ be a labelled Sierpinski Graph of level $l$ with $n_l$ vertices.
Then $r_{\textrm{ind}} (S_l) \leq n_l^{c_d}$, where $c_d$ is a positive constant, independent from $n_l$ and, thus, from $l$.
\label{corollary-pol}
\end{corollary}

\subsection{The emergence of organized subgraphs in evolving graphs}
Although random graphs (see e.g.~\cite{bollobas85}) are a powerful tool for proving limit properties, i.e. properties that hold in the limit as the graph size grows asymptotically, its main limitation is that it does not say anything about {\em specific}, finite, graph instances, for a fixed size $n$.
Kolmogorov complexity theory, on the other hand, focuses on the study of specific, finite, objects (graphs in our case). Based on this theory, we can prove the following:
\begin{theorem}
Let $S_l$ be a labelled Sierpinski Graph on $n_l$ vertices. Let $G$ be a labelled incompressible graph on $n$ vertices that contains $S_l$ as an induced subgraph. Then $n  \geq 2^{\frac{n_l-1}{2}}$.
\label{sierpinski-theorem}
\end{theorem}

\begin{proof}
Let $G$ to be a labelled {\em incompressible} graph with $n$ vertices whose encoding $E(G)$ as a binary string (see Definition~\ref{def.gc}) has length $l(E(G))$. Since $G$ is incompressible, it holds
\begin{equation}
C(E(G)|n,P) \geq \frac{n(n-1)}{2}
\label{incomp}
\end{equation}
where $P$ is a program that can reconstruct $E(G)$ from the value $n$
and an alternative encoding $E'(G)$ as described below.

Since $G$ contains, from our assumption, $S_l$ as an {\em induced} subgraph, we can describe $G$ by forming an alternative encoding $E'(G)$ constructed from $E(G)$ as follows:
\begin{enumerate}
\item We add to the encoding $E(G)$ the description of the Sierpinski Graph subgraph of $G$. In order to describe such a graph of $n_l$ vertices in the graph $G$ of $n$ vertices, we need $\log{n \choose n_l}$ bits to denote the subset of the $n_l$ vertices and $\log n_l!$ bits to denote the specific ordering $i_1i_2 \ldots i_{n_l-1}i_{n_l}$ that provides the structure of a Sierpinski Graph, as we see it in the example graph of Figure~\ref{pattern2}, i.e. we form the Sierpinski Graph on $n_l$ vertices and label them from top to bottom and from left to right (e.g. see Figure~\ref{pattern2} for an indication of the ``top'' and the ``bottom'' of a Sierpinski Graph).
In total, to describe this Sierpinski Graph subgraph, we need at most $\log{n \choose n_l} +\log n_l!$ bits and, since ${n \choose n_l} \leq \frac{n^n_l}{n_l!}$, we need, at most
\begin{equation}
\log \frac{n^n_l}{n_l!} + \log n_l! = n_l \log n
\label{upzigzag}
\end{equation}
bits.
\item We delete, from $E(G)$, {\em all} the bits that encode the
edges of the Sierpinski Graph subgraph, saving $\frac{n_l(n_l-1)}{2}$ bits.
\end{enumerate}
Now, it is easy to provide an algorithm $P$ that, given $E'(G)$, constructs the Sierpinski Graph subgraph on $n_l$ vertices given as input the number $n$ of vertices of the graph, the value of $n_l$ and the specific ordering
$i_1i_2 \ldots i_{n_l-1}i_{n_l}$, of the vertices of the Sierpinski Graph.

The length of the new encoding is, at most
\begin{equation}
l(E'(G)) = l(E(G)) +  n_l \log n - \frac{n_l(n_l-1)}{2}.
\label{altencoding}
\end{equation}
Given the value of $n$, the program $P$ can reconstruct $E(G)$ from $E'(G)$. Thus
\begin{equation}
C(E(T)|n,P) \leq l(E'(T)).
\label{incomp2}
\end{equation}
Since $G$ is incompressible, it must hold $l(E'(G)) \geq l(E(G))$. From~(\ref{altencoding}), this can, only, hold if
\begin{equation}
n_l \log n \geq  \frac{n_l(n_l-1)}{2} \Leftrightarrow n \geq 2^{\frac{n_l-1}{2}}
\label{kolm-ineq}
\end{equation}
which is the required. \hfill $\square$
\end{proof}
Theorem~\ref{sierpinski-theorem} states that no incompressible graph with {\em fewer} than $2^\frac{n_l-1}{2}$ vertices can contain $S_l$ as an {\em induced} subgraph. Consequently, an incompressible graph cannot contain $S_l$ as a {\em monochromatic} induced subgraph either, in {\em any} two colouring of its edges.
Thus, from Corollary~\ref{corollary-pol} and Theorem~\ref{sierpinski-theorem}, we have the following:
\begin{theorem}
No incompressible graph on $r_{\textrm{ind}} (S_l)$ vertices can contain $S_l$ as an induced subgraph except, possibly, for a finite set of values for $l$.
\label{zig-zag-corol}
\end{theorem}
\begin{proof}
Let $G$ be an incompressible graph on $n$ vertices that contains $S_l$ as an induced subgraph.
%From the definition of $r_{\textrm{ind}} (Z_k)$, it should hold $n \geq r_{\textrm{ind}} (Z_k)$.
Since $r_{\textrm{ind}} (S_l) \leq n_l^{c_d}$ from Corollary~\ref{corollary-pol}, it follows that $n \leq n_l^{c_d}$. Also, the bound $n \geq  2^{\frac{n_l-1}{2}}$ holds from
Theorem~\ref{sierpinski-theorem}. Thus, it follows that $n \leq n_l^{c_d}$ can only hold for a finite set of values for $n_l$, whose cardinality depends on the constant $c_d$, since the growth rate of $n$ is exponential in $n_l$ while the growth rate of the bound $n_l^{c_d}$ for $r_{\textrm{ind}} (S_l)$ is polynomial in $n_l$. \hfill $\square$
\end{proof}
Moreover, based on Lemma~\ref{lem.frac}, we obtain the following {\em stronger}, than Theorem~\ref{zig-zag-corol}, result:
\begin{theorem}
Almost all graphs on $r_{\textrm{ind}} (S_l)$ vertices (a fraction of $(1-\frac{1}{r_{\textrm{ind}} (S_l)})$ of them) are such that no two-colouring of their edges contains an induced monochromatic copy
of $S_l$, as $l$ grows.
\label{thm_fraction}
\end{theorem}
\begin{proof}
Following the same line of proof as in Theorem~\ref{sierpinski-theorem}, we now start with a labelled graph $G$ such that
\begin{equation}
C(E(G)|n,P) \geq \frac{n(n-1)}{2} - \log{n}.
\label{incomp3}
\end{equation}
These graphs form a fraction of at least $(1 - \frac{1}{n})$ of all labelled graphs on $n$ vertices, according to Lemma~\ref{def-lemma}.
For $n = r_{\textrm{ind}} (S_l)$ these graphs form a ratio of at least $(1 - \frac{1}{r_{\textrm{ind}} (S_l)})$ of all graphs with $r_{\textrm{ind}} (S_l)$ vertices. What is stated below, applies to {\em all} of these graphs
which, as $l$ tends to infinity, including {\em almost all} possible graphs on $r_{\textrm{ind}} (S_l)$ vertices.

The rest of the proof follows closely the proof of Theorem~\ref{sierpinski-theorem}, setting $n = r_{\textrm{ind}} (S_l)$, but now the following inequality must be satisfied instead of Inequality~(\ref{kolm-ineq}):
\begin{equation}
n_l \log n + \log n \geq  \frac{n_l(n_l-1)}{2} \Leftrightarrow n \geq 2^{\frac{n_l(n_l-1)}{2(n_l+1)}}
\Leftrightarrow r_{\textrm{ind}} (S_l) \geq 2^{\frac{n_l(n_l-1)}{2(n_l+1)}}.
\label{kolm-ineq2}
\end{equation}
However, since $r_{\textrm{ind}} (S_l) \leq n_l^{c_d}$ for some constant $c_d$ depending only on the maximum degree $d$ of the graph vertices,
according to Theorem~\ref{upper-induced}, Inequality~(\ref{kolm-ineq2}) would require $n_l^{c_d} \geq 2^{\frac{n_l(n_l-1)}{2(n_l+1)}}$, which cannot hold from some value of $l$ onwards. \hfill$\Box$
\end{proof}
%
%Since $r_{\textrm{ind}} (Z_k)$
%
%
\begin{definition} [Size Constructible Graphs]
Let ${\cal F}$ be a family of graphs. We call ${\cal F}$ size constructible if each of the graphs in ${\cal F}$ can be uniquely constructed by an algorithm $P_{\cal F}$ which takes as inputs the graph's size, i.e. the number of vertices of the graph, and, possibly, a permutation that gives some ordering information about the vertices.
\label{SCG}
\end{definition}
A direct consequence of Definition~\ref{SCG} is that the family of Sierpinski Graphs is size constructible. This is due to the fact that a Sierpinski Graph $S_l$ with $n_l$ vertices can be described with only information $n_l$ and the permutation that denotes the ordering of its vertices which the reconstruction algorithm $P$ that we described in the proof of Theorem~\ref{sierpinski-theorem} uses in order to reconstruct the graph.

Then, Theorem~\ref{sierpinski-theorem} can be generalized as follows:
\begin{theorem}
Let ${\cal F}$ be a family of size constructible graph. Let Let $G$ be a labelled incompressible graph on $n$ vertices that contains  as an induced subgraph a graph in ${\cal F}$ with $k$ vertices. Then $n \geq 2^{\frac{k-1}{2}}$ if $P_{\cal F}$ requires information about the ordering of the vertices and $n \geq k 2^{\frac{k}{2}} \left(\frac{1}{e \sqrt{2}} - o(1) \right)$ if $P_{\cal F}$ does not require this information.
\label{SCGsierpinski-theorem}
\end{theorem}
\begin{proof}
If the algorithm $P_{\cal F}$ needs, except from $k$, also the ordering of the vertices, then a graph in ${\cal F}$ can be described using $k \log n + \log k!$ bits, as in the case of Sierpinski Graphs in Theorem~\ref{sierpinski-theorem}, thus the proof of this theorem also applies to the case of ${\cal F}$.

If $P_{\cal F}$ does not need the ordering information, then instead of needing
the number of bits given in Equation~(\ref{upzigzag}) in the proof of
Theorem~\ref{sierpinski-theorem}, we need, at most
\begin{equation}
\log \frac{n}{k!}  = k \log n - \log k!.
\label{upzigzag2}
\end{equation}
bits since the ordering information, i.e. $\log k!$ bits, is not required.

Thus, the following inequality must hold for an incompressible graph $G$, as in
Inequality~(\ref{kolm-ineq}) for the Sierpinski Graphs in the proof of Theorem~\ref{sierpinski-theorem}):
\begin{equation}
k \log n - \log k! \geq  \frac{k(k-1)}{2}.
\label{kolm-ineq3}
\end{equation}
Using Stirling's approximation $k! \approx (\frac{k}{e})^k \sqrt{2\pi k}$ and
solving for $n$,
%z_{\cal F}(k,k)$,
Inequality~(\ref{kolm-ineq3}) leads to
$$
n \geq k 2^{\frac{k}{2}} \left(\frac{1}{e \sqrt{2}} - o(1) \right)
$$
which is the required. \hfill $\square$
\end{proof}
In our context, i.e. the emergence of an organized, close-knit, community in the form of a
Sierpinski Graph, the results above have certain ineresting consequences.
First of all, Theorem~\ref{sierpinski-theorem} and Theorem~\ref{SCGsierpinski-theorem}
(the generalization of Theorem~\ref{sierpinski-theorem}) show that the existence or purposeful
formation of
orgagnized structures, such as the Sierpinski Graphs, in
{\em incompressible}, i.e. random-like, networks of interacting agents requires
the networks to be {\em exponentially large} with respect to the size of the organized structure.
Smaller interconnection networks, i.e. {\em polynomially large} in the size of Sierpinski Graph,
almost certainly do not contain such organized structures. 

\section{Induced Ramsey numbers for incompressible graphs}

In this section we investigate Ramsey Numbers in the context of Kolmogorov random graphs. More specifically, we prove bounds on the size of Kolmogorov Random graphs so as to contain a Sierpinski Graph (or other size constructible graph) as a subgraph.

\begin{definition} (Induced Ramsey numbers for incompressible graphs)
The {\it induced Ramsey number} for incompressible graphs, $r^{\textrm{{\tiny INC}},\delta(n)}_{\textrm{ind}} (H)$, is defined as the {\em minimum} integer
for which there is a Kolmogorov Random graph $G$ with randomness
deficiency at least $\delta(n)$ on $r^{\textrm{{\tiny INC}},\delta(n)}_{\textrm{ind}} (H)$ vertices such that every 2-colouring of $E(G)$ contains an induced monochromatic copy of $H$.
\label{KR_def-induced-ramsey}
\end{definition}

\begin{theorem}
Let $H$ be a size constructible graph on $k$ vertices and of maximum degree at most $d$, whose description needs vertex ordering information (a similar result holds for size constructible graphs that do not require such information).
Let $n_1 = 2^{\frac{k-1}{2}} - 1$ and $n_2 = r_{\textrm{ind}} (H)$.
Then for  $\delta(n) \geq n_1 n_2 + {n_2 \choose 2}$, it holds that (i) 
$r^{\textrm{{\tiny INC}},\delta(n)}_{\textrm{ind}}(H)$ exists, and (ii)
$r^{\textrm{{\tiny INC}},\delta(n)}_{\textrm{ind}}(H) < 2^{\frac{k-1}{2}} + r_{\textrm{ind}} (H)$.
%for some constant $c_d$ depending only on $d$.
\label{KR_theorem}
\end{theorem}

\begin{proof}
Let $G_1$ be any {\em incompressible} graph on $n_1 = 2^{\frac{k-1}{2}} - 1$ vertices, so that (according to Theorem~\ref{SCGsierpinski-theorem}) it does not contain $H$ as an induced
subgraph. Since $G_1$ is incompressible, it holds $C(E(G_1)) \geq {n_1 \choose 2}$.

Let, also, $G_2$ be any graph on $n_2 = r_{\textrm{ind}} (H) \leq k^{c_d}$
vertices such that for every two-colouring of its edges it contains an induced
monochromatic copy of $H$ (the existence of such a graph is guaranteed by Theorem~\ref{upper-induced}).
%Given $n_2$, we tune
%$n_1$ so that $2^{\frac{k-1}{2}} \leq n_1 + n_2$.
%
%
%Due to its minimality, we can assume that $G_2$ is {\em connected}.
%
%Thus, for $n_1$ and $n_2$, the following inequalities hold, by construction:
%
%\begin{equation}
%2^{\frac{k-1}{2}} \leq n_1 + n_2 < 2^{\frac{k-1}{2}} + k^{c_d}.
%\label{assumption}
%\end{equation}
%
%
Let $n=n_1 + n_2$. We focus on the graphs $G$ with $n$ vertices which,
simply, consist of one copy of $G_1$ and one copy of $G_2$, that is,
two subgraphs isomorphic to $G_1$ and $G_2$. There are no additional edges except those in these two subgraphs.
% In other words, we focus on graphs $G$ such that $G = G_1 \cup G_2$.

Obviously, every two-colouring of $E(G)$ contains an induced monochromatic copy of $H$, since $G_2$ does. This proves that 
$r^{\textrm{{\tiny INC}},\delta(n)}_{\textrm{ind}}(H)$ exists
for Kolmogorov random, i.e. incompressible, graphs (at least for the deficiency function $\delta(n)$ which will be defined below).

Also, $G$ has $n = n_1 + n_2 = 2^{\frac{k-1}{2}} - 1 + r_{\textrm{ind}} (H)$ vertices.
We will show that for appropriate randomness deficiency function $\delta(n) = \delta(n_1,n_2)$ the following holds, i.e. $G$ is $\delta(n_1,n_2)$-incompressible:
\begin{eqnarray}
& & C(E(G)) \geq {n_1 + n_2 \choose 2} - \delta(n_1,n_2) \nonumber \\
& &  \mbox{with }  \delta(n_1,n_2) \geq {n_2 \choose 2} +
n_1n_2.
\label{KRG}
\end{eqnarray}
Assume, towards a contradiction, that $C(E(G)) < {n_1 + n_2 \choose 2} - \delta(n_1,n_2)$
for a randomness deficiency function $\delta(n_1,n_2)$ that obeys the inequality in~(\ref{KRG}).
Then the following holds:
\begin{eqnarray}
C(E(G)) & < & {n_1 + n_2 \choose 2} - \delta(n_1,n_2) \leq  \nonumber \\
& & {n_1 + n_2 \choose 2} - \left [{n_2 \choose 2} +
n_1n_2 \right] = {n_1 \choose 2}.
\label{KRGbound}
\end{eqnarray}
We will describe an algorithm that can reconstruct $G_1$ from a description of
$G$ of $C(E(G)) < {n_1 \choose 2}$ bits contradicting
our assumption that $C(E(G_1)) \geq {n_1 \choose 2}$, i.e. that $G_1$ is incompressible.

Let us assume that we have a description of $G$ of $C(E(G)) < {n_1 \choose 2}$ bits.
Then we can reconstruct $G_1$ as follows. We first reconstruct $G$ using its description of $C(E(G))$ bits.
We, then, compute its {\em connected components} using, e.g.,
a {\em Depth First Search (DFS)} algorithm (see~\cite{algorithms}).
Let these components be $G_{i_1}, G_{i_2}, \ldots, G_{i_s}$.
One of these components must be the graph $G_2$. The union of the rest of the components should be the graph $G_1$. However, it is possible that only {\em one} component exists, beyond $G_2$, if $G_1$ is a {\em connected} graph.
Our next step is to identify $G_2$. If we succeed in this task, then
we have identified $G_1$: it is the graph that is composed of vertices and edges of $G$ which do not belong in $G_2$.

In order to identify $G_2$ we work on the components $G_{i_1}, G_{i_2}, \ldots, G_{i_s}$, one at a time.
We should be cautious here, since $G_2$ may not be a {\em connected} graph and, thus, we need to locate {\em all} its components.
The crucial observation is that $G_2$ cannot contain a component such that no two-colouring of its edges contains an induced monochromatic copy of $H$. Otherwise, we could dispense with this component and have a smaller graph satisfy the definition of $r_{\textrm{ind}} (H)$, contradicting its {\em minimality} requirement.

Let $G_{i_l}$ be the currently examined component. We produce all possible
edge 2-colourings of this component and for {\em each} of them we check
whether $G_2$ contains a monochromatic copy of $H$.
Observe that this can be true only for $G_2$ components
since all other components belong to $G_1$, which is an incompressible graph
%(as components of an incompressible graph, the graph $G_1$)
with less than $2^{\frac{k-1}{2}}$ vertices and thus, according to
Theorem~\ref{sierpinski-theorem}, its edge 2-colourings (and, consequently,
its components' 2-colourings) cannot contain a monochromatic copy of $H$.
Thus, having identified $G_2$ we can reconstruct $G_1$ as the
subgraph of $G$ containing the rest of the components.
In this way, we have managed to reconstruct $G_1$ using less than ${n_1 \choose 2}$ bits, which
is a contradiction. Consequently,
$C(E(G)) \geq {n_1 + n_2 \choose 2} - \delta(n_1,n_2)$, that is $G$ is $\delta(n_1,n_2)$-incompressible.

In conclusion, $G$ is a graph of $n = n_1 + n_2$
which, for $\delta(n) \geq \delta(n_1,n_2)$, is $\delta(n)$-incompressible and all the 2-colourings of its edges contain an induced monochromatic copy of $H$. Thus, it follows
that $r^{\textrm{{\tiny INC}},\delta(n)}_{\textrm{ind}}(H)
\leq n = 2^{\frac{k-1}{2}} - 1 + r_{\textrm{ind}} (H)$ or, in
simpler form, $r^{\textrm{{\tiny INC}},\delta(n)}_{\textrm{ind}}(H) < 2^{\frac{k-1}{2}} + r_{\textrm{ind}} (H)$.
\hfill $\square$
\end{proof}
From Theorems~\ref{sierpinski-theorem} and~\ref{KR_theorem}, the following is derived:
\begin{corollary}
For every graph $H$ on $k$ vertices and maximum degree at most $d$, it holds that
$2^{\frac{k-1}{2}} \leq r^{\textrm{{\tiny INC}},\delta(n)}_{\textrm{ind}}(H) < 2^{\frac{k-1}{2}} + r_{\textrm{ind}} (H)$.
%for some constant $c_d$ depending only on $d$.
\label{cor_KR_theorem}
\end{corollary}
Note that, since $n = n_1 + n_2$, the canonical representation of $G$ has $|E(G)|={n_1 + n_2 \choose 2}$ bits while
we set $\delta(n) = \delta(n_1, n_2) = n_1 n_2 + {n_2 \choose 2}$.
Thus, considering $k$ as a varying parameter, $\delta(n)$
is in the order of $\sqrt{|E(G)|}$. We could not prove
Theorem~\ref{KR_theorem} for deficiency functions $\delta(n)$ smaller than this, e.g. $\delta(n) = \log{|E(G)|}$.
However, we believe that this is not possible since when
a graph contains some regularity, in our case the Sierpinski
Graph $S_l$ on $n_l$ vertices, its complexity drops since the regularity can be deployed in reducing the size of its description.

\section{Conclusions and directions for further research}

Society is a complex human creation composed of multifaceted autonomous,
interacting, agents and their interrelationships.
It is a common theme in numerous research works that out
of evolving (even randomly) interactions and relationships
important phenomena and patterns emerge in complex networks of agents.
Such results rely on well established mathematical and physical
theories about models of interacting agents,
such as Random Graph theory and Complex Network theory.

In this paper we considered two well established mathematical theoretical frameworks targeting the concept of complexity of finite objects, the {\em society and its interacting agents} in our case: {\em Kolmogorov Complexity}
and {\em Ramsey Theory}. Both of these theories, each from
another perspective, study the conditions upon which
certain substructures appear (or do not appear) in large, evolving,
structures such as societal networks of interacting agents,
also deriving estimates on how large the structures
should become in order to contain such regularities.
We applied elements of these theories in order to study the appearance of
regular structures or patterns, in evolving societies, that have certain desirable properties. One of these properties, which was among our targets, is {\em close-knittedness} or, in other words, the property that describes structures of communicating agents whose members interact closely with each other and defend, strongly, their group's coherence and views against the containing, larger, structure.

More specifically, with Theorem~\ref{zig-zag-corol} we proved that the Sierpinski Graph,
as a highly organized structure, cannot emerge in evolving graphs (societies) unless they reach a sufficient size, exponential in the size of the organized structure.
With Theorem~\ref{KR_theorem} and Corollary~\ref{cor_KR_theorem} we gave bounds on
the size the graphs (societies) that contains with certainty any organized
structure of bounded degree (i.e. relationships) among its vertices (i.e. society members).

We hope that our work will contribute to the further exploitation of the rich
mathematical theories of complex structures and their long-term evolutionary properties.
In this way, we feel we can strengthen the efforts towards the study of Social Sciences
with methodologies stemming from exact sciences and formal systems.

\end{document}